\documentclass[12pt]{article}
\usepackage{graphics,amsmath,amssymb}
\usepackage{amsthm}
\usepackage{amsfonts}
\usepackage{latexsym}
\usepackage{amscd}
\usepackage{atbegshi}
\AtBeginDocument{\AtBeginShipoutNext{\AtBeginShipoutDiscard}}

\newtheorem{theorem}{Theorem}[section]
\newtheorem{corollary}{Corollary}[section]

\begin{document}
\begin{center}
\title{A Radix Representation for each van der Waerden number $W(r, k)$ with $r$ colors: Why \(\log_{r}W(r, k) < k^{2}\) is true whenever $k$ is the number of terms in the arithmetic progression}
\end{center}
\maketitle
\begin{center}
\author{\textbf{Robert J. Betts}}\\
\emph{The Open University\\Postgraduate Department of Mathematics and Statistics~\footnote{2012--2013.}\\ (Main Campus) Walton Hall, Milton Keynes, MK7 6AA, UK\\
Robert\_Betts@alum.umb.edu}
\end{center}
\begin{abstract}
Here we show that by expressing a van der Waerden number $W(r, k)$ by its radix polynomial representation, it not only is possible to locate each proper subset on $\mathbb{R}$ in which the van der Waerden number lies, but also to show that conditions exist for which the logarithm of the van der Waerden number necessarily is bounded above by the square of the number of terms $k$ in the arithmetic progression. Furthermore we also use the method to find a mathematical expression or formula for the ratio of two ``consecutive" van der Waerden numbers of the kind $W(r, k)$, $W(r, k + 1)$.,,
\end{abstract}
\section{Some a priori information about van der Waerden numbers}
Any integer, and this includes any van der Waerden number $W(r, k)$ considered as an integer~\cite{Graham1},~\cite{Graham and Rothschild, Graham and Spencer},~\cite{Landman and Robertson, Landman and Culver}~\cite{Khinchin, van der Waerden}, where $r$ is the number of colors and $k$ is the number of terms in the arithmetic progression, has a radix polynomial expression for it~\cite{Abramowitz, Rosen},~\cite{Bettsa, Bettsb}. For instance 
\begin{equation}
W(r, k) = b_{n}r^{n} + b_{n - 1}r^{n - 1} + \cdots + b_{0} \in [r^{n}, r^{n + 1})
\end{equation}
is true always for some integer \(b_{n} \in [1, r - 1]\), for some additional integers \(b_{n - 1}, \ldots, b_{0} \in [0, r - 1]\) and where
\begin{equation}
n = \lfloor \log_{r}W(r, k)\rfloor. 
\end{equation} 
So given the integer $r$, each van der Waerden number $W(r, k)$ is shown this way to lie always in some interval on $\mathbb{R}$ of the form $[r^{n}, r^{n + 1})$ where $n$ is given by Eq. (2). Put another way, each van der Waerden number $W(r, k)$, is bounded above and below on $\mathbb{R}$ as
$$
r^{n} \leq W(r, k) < r^{n + 1},
$$
while each $\log_{r}W(r, k)$ is bounded above and below on $\mathbb{R}$ as \(n \leq \log_{r}W(r, k) < n + 1\). On the other hand when $k$ is chosen as the radix (See Section 2) these bounds become, respectively, \(k^{m} \leq W(r, k) < k^{m + 1}\) and \(m \leq \log_{k}W(r, k) < m + 1\), where \(m = \lfloor \log_{k}W(r, k)\rfloor\) and where
$$
W(r, k) \in [r^{n}, r^{n + 1}) \cap [k^{m}, k^{m + 1}) \Longrightarrow [r^{n}, r^{n + 1}) \cap [k^{m}, k^{m + 1}) \not = \emptyset.
$$
\indent We certainly hope the reader is aware and \emph{should be aware}, that in Eq. (1), the value of the radix polynomial representation for the integer $W(r, k)$ \emph{never exceeds the value}
\begin{equation}
r^{n + 1} - 1,
\end{equation}
which when expanded outright is
$$
r^{n + 1} - 1 = (r - 1)(r^{n} + r^{n - 1} + r^{n - 2} + \cdots + 1).
$$
All this enables us to show that conditions exist so that both the integer $n + 1$ and the real number $\log_{r}W(r, k)$ are bounded above necessarily and always by $k^{2}$, as now we prove (See also the Tables at the end of this section).
\begin{theorem}
Let both $W(r, k)$ and $W(r, k^{\prime})$, be any two van der Waerden numbers that satisfy the following conditions:
\begin{enumerate}
\item \(W(r, k) > W(r, k^{\prime})\) and \(k > k^{\prime}\) always, where
\begin{equation}
W(r, k) = b_{n}r^{n} + b_{n - 1}r^{n - 1} + \cdots + b_{0} \in [r^{n}, r^{n + 1}).
\end{equation}
\\
\item \(\log_{r} W(r, k^{\prime}) < \lfloor \log_{r} W(r, k)\rfloor = n\).\\
\item \(k^{\prime} \in [3, \sqrt{n + 1})\) is true always for all $k^{\prime}$.
\end{enumerate}
Then \(W(r, k) < r^{n + 1} \leq r^{k^{2}}\) is true for all \(k \geq \sqrt{n + 1}\).
\end{theorem}
\begin{proof}
By Condition 1 and Condition 3, $k^{\prime}$ can assume any integer value from three to any integer less than $\sqrt{n + 1}$, where
$$
l.u.b([3, \sqrt{n + 1}) = \sqrt{n + 1} \not \in [3, \sqrt{n + 1}).
$$ 
Moreover \(k \in [3, \sqrt{n + 1})\) is impossible since, by Condition 1, we are given that \(k > k^{\prime}\) is true always while by Condition 3, \(k^{\prime} \in [3, \sqrt{n + 1})\) is true always and this interval on $\mathbb{R}$ only can contain a finite number of integer values for $k^{\prime}$ for each integer exponent $n$. For instance if $k^{\prime}$ assumes the largest integer value possible in this set then we have still that \(k > k^{\prime}\), which means \(k \not \in [3, \sqrt{n + 1})\).\\
\indent Then since by Condition 1 \(k > k^{\prime}\) it must be that 
\begin{equation}
k \in [3, \infty) - [3, \sqrt{n + 1}) = [\sqrt{n + 1}, \infty), 
\end{equation}
that is, the integer $k$ lies always in $[\sqrt{n + 1}, \infty)$, a set that is the relative complement in $[3, \infty)$ of the set $[3, \sqrt{n + 1})$  whenever \(k^{\prime} \in [3, \sqrt{n + 1})\). Then since 
\begin{equation}
g.l.b.([\sqrt{n + 1},\infty)) =  \sqrt{n + 1},
\end{equation}
we get \(k \geq \sqrt{n + 1}\), from which we obtain both by this fact and from Eq.(4) in Condition 1,
\begin{equation}
W(r, k) < r^{n + 1} \leq r^{k^{2}} \: \forall k \geq \sqrt{n + 1},
\end{equation}
since \(\sqrt{n + 1} \leq k \Longrightarrow n + 1 \leq k^{2}\). 
\end{proof}
Since \(n \leq \log_{r}W(r, k) < n + 1\), the Theorem also shows that
\begin{equation}
n \leq \log_{r}W(r, k) < n + 1 \leq k^{2},
\end{equation}
or, \(\log_{r}W(r, k) < k^{2}\), as was mentioned in the Abstract.

\newpage
\begin{center}
\begin{tabular}{|l      |c      |c          |c                         |c                                        |r|}
\hline
               $r$   &  $k$  &  $n$   &      \(W(r, k) = N\)     &      \(N = r^{\log_{r}W(r, k)}\)     &        $r^{n}$     \\   
\hline
                $2$  &  $3$  &  $3$   &      \(W(2, 3) = 9\)     &       \(9 = 2^{3.17010\ldots}\)      &        $2^{3}$    \\
 
                $2$  &  $4$  &  $5$   &      \(W(2, 4) = 35\)    &       \(35 = 2^{5.12963\ldots}\)     &        $2^{5}$   \\
                
                $2$  &  $5$  &  $7$   &      \(W(2, 5) = 178\)   &       \(178 = 2^{7.47623\ldots}\)    &        $2^{7}$  \\ 

                $2$  &  $6$  &  $10$  &      \(W(2, 6) = 1132\)  &       \(1132 = 2^{10.14534\ldots}\)  &        $2^{10}$   \\
                 
                $3$  &  $3$  &  $3$   &      \(W(3, 3) = 27\)    &       \(27 = 3^{3.00002\ldots}\)     &        $3^{3}$   \\
                
                $3$  &  $4$  &  $5$   &      \(W(3, 4) = 293\)   &       \(293 = 3^{5.17037\ldots}\)    &        $3^{5}$    \\

                $4$  &  $3$  &  $3$   &      \(W(4, 3) = 76\)    &       \(76 = 4^{3.12417\ldots}\)     &        $4^{3}$    \\
  
\hline
\end{tabular}
\end{center}
\begin{center}
Table 1.
\end{center}
\begin{center}
\begin{tabular}{|l      |c      |c                  |c      |c            |c              |c                  |c              |c              |r|}
\hline
                $r$  &  $k$  &  $\sqrt{n + 1}$  &   $n$  &  $\log r$  &   $\log k$    &   $r^{n}$     &       $W(r, k)$   &   $r^{n + 1}$  &  $r^{k^{2}}$ \\   
\hline
                $2$  &  $3$  &  $2$             &   $3$  &  $0.6931$  &   $1.0986$    &   $2^{3}$     &       $9$         &   $2^{4}$      &  $2^{9}$  \\
 
                $2$  &  $4$  &  $2.449\ldots$   &   $5$  &  $0.6931$  &   $1.3862$    &   $2^{5}$     &       $35$        &   $2^{6}$      &  $2^{16}$  \\
                
                $2$  &  $5$  &  $2.828\ldots$   &   $7$  &  $0.6931$  &   $1.6094$    &   $2^{7}$     &       $178$       &   $2^{8}$      &  $2^{25}$ \\ 

                $2$  &  $6$  &  $3.316\ldots$   &   $10$ &  $0.6931$  &   $1.7917$    &   $2^{10}$    &       $1132$      &   $2^{11}$     &  $2^{36}$  \\
                 
                $3$  &  $3$  &  $2$             &   $3$  &  $1.0986$  &   $1.0986$    &   $3^{3}$     &       $27$        &   $3^{4}$      &  $3^{9}$  \\
                
                $3$  &  $4$  &  $2.449\ldots$   &   $5$  &  $1.0986$  &   $1.3862$    &   $3^{5}$     &       $293$       &   $3^{6}$      &  $3^{16}$  \\

                $4$  &  $3$  &  $2$             &   $3$  &  $1.3862$  &   $1.0986$    &   $4^{3}$     &       $76$        &   $4^{4}$      &  $4^{9}$  \\
  
\hline
\end{tabular}
\end{center}
\begin{center}
Table 2. Logarithms taken to the base $e$.
\end{center}
\section{The rational number $\frac{W(r, k + 1)}{W(r, k)}$, when $k$ is large}
In this section we find an expression for the rational number~\cite{Bettsb},
\begin{equation}
\frac{W(r, k + 1)}{W(r, k)},
\end{equation}
when $k$ is large.
\begin{corollary}
Let
\begin{equation}
W(r, k) = c_{m_{k}}k^{m_{k}} + c_{m_{k} - 1}k^{m_{k} - 1} + \cdots + c_{0, m_{k}},
\end{equation}
be the radix polynomial representation for $W(r, k)$ with radix $k$, and let 
\begin{equation}
W(r, k + 1) = c_{m_{k + 1}}(k + 1)^{m_{k + 1}} + c_{m_{k + 1} - 1}(k + 1)^{m_{k + 1} - 1} + \cdots + c_{0, m_{k + 1}}
\end{equation}
be the radix polynomial representation for $W(r, k + 1)$ with radix $k + 1$~\cite{Bettsb}. Furthermore let \(\alpha \in \mathbb{N}\cup \{0\}\) be such that \(k = r \pm \alpha\), meaning \(k = r + \alpha\) when \(k > r, \alpha > 0\) and \(k = r - \alpha\) whenever \(k\leq r, r > \alpha \geq 0\). Then for large $k$,
\begin{eqnarray}
\frac{W(r, k + 1)}{W(r, k)}&=&k^{m_{k + 1} - m_{k}}\frac{c_{m_{k + 1}}}{c_{m_{k}}}(1 + o(1))\\
                           &=&r^{m_{k + 1} - m_{k}}\left(1 \pm \frac{\alpha}{r}\right)^{m_{k + 1} - m_{k}}\frac{c_{m_{k + 1}}}{c_{m_{k}}}(1 + o(1)).
\end{eqnarray}
\end{corollary}
\begin{proof}
\begin{eqnarray}
& &\frac{W(r, k + 1)}{W(r, k)}\\
&=&\frac{c_{m_{k + 1}}(k + 1)^{m_{k + 1}} + c_{m_{k + 1} - 1}(k + 1)^{m_{k + 1} - 1} + \cdots + c_{0, m_{k + 1}}}{c_{m_{k}}k^{m_{k}} + c_{m_{k} - 1}k^{m_{k} - 1} + \cdots + c_{0, m_{k}}}\nonumber\\
&=&\frac{(k + 1)^{m_{k + 1}}}{k^{m_{k}}} \cdot \frac{c_{m_{k + 1}} + \frac{c_{m_{k + 1} - 1}}{k + 1} + \cdots + \frac{c_{0, m_{k + 1}}}{(k + 1)^{m_{k + 1}}}}{c_{m_{k}} + \frac{c_{m_{k} - 1}}{k} + \cdots + \frac{c_{0, m_{k}}}{k^{m_{k}}}}\\
&=&k^{m_{k + 1} - m_{k}}\left(1 + \frac{1}{k}\right)^{m_{k + 1}}\cdot \frac{c_{m_{k + 1}} + \frac{c_{m_{k + 1} - 1}}{k + 1} + \cdots + \frac{c_{0, m_{k + 1}}}{(k + 1)^{m_{k + 1}}}}{c_{m_{k}} + \frac{c_{m_{k} - 1}}{k} + \cdots + \frac{c_{0, m_{k}}}{k^{m_{k}}}}\\
&=&k^{m_{k + 1} - m_{k}}(1 + o(1))\frac{c_{m_{k + 1}} + o(1)}{c_{m_{k}} + o(1)},
\end{eqnarray}
when $k$ is so large that \(\frac{1}{k} = o(1)\). With the use of \(k = r \pm \alpha\) and a little rearranging in Eqn. (17) we get
\begin{eqnarray}
\frac{W(r, k + 1)}{W(r, k)}&=&k^{m_{k + 1} - m_{k}}\frac{c_{m_{k + 1}}}{c_{m_{k}}}(1 + o(1))\\
                           &=&r^{m_{k + 1} - m_{k}}\left(1 \pm \frac{\alpha}{r}\right)^{m_{k + 1} - m_{k}}\frac{c_{m_{k + 1}}}{c_{m_{k}}}(1 + o(1)),
\end{eqnarray}
where in Eqn. (19) one chooses \(1 + \frac{\alpha}{r}\) when \(k > r, \alpha > 0\) and \(1 - \frac{\alpha}{r}\) when \(k \leq r, r > \alpha \geq 0\).
\end{proof}
One sees in Eqn. (19) that for large $k$, the rational number
\begin{equation}
\frac{W(r, k + 1)}{W(r, k)},
\end{equation}
has $r$ as a factor. So how close are the two exponents $m_{k + 1}$ and $m_{k}$? Since~\cite{Bettsb},
\begin{eqnarray}
W(2, 3)&=&9 = 3^{2} \in [3^{2}, 3^{3}),\\
W(2, 4)&=&35 = 2\cdot 4^{2} + 3 \in [4^{2}, 4^{3}),\nonumber\\
W(2, 5)&=&178 = 5^{3} + 2\cdot 5^{2} + 3 \in [5^{3}, 5^{4}),\\
W(2, 6)&=&1132 = 5\cdot 6^{3} + 1\cdot 6^{2} + 2\cdot 6^{1} + 4 \in [6^{3}, 6^{4}),\nonumber\\
W(3, 3)&=&27 = 3^{3} \in [3^{3}, 3^{4}),\\
W(3, 4)&=&293 = 4^{4} + 2\cdot 4^{2} + 4^{1} + 1 \in [4^{4}, 4^{5}),\\ 
W(4, 3)&=&76 = 2\cdot 3^{3} + 2\cdot 3^{2} + 3^{1} + 1 \in [3^{3}, 3^{4}),
\end{eqnarray} 
we see that \(m_{k + 1} - m_{k} \in \{0, 1\}\) is true in Eqn. (19) at the very least, for the four rational numbers
\begin{eqnarray}
\frac{35}{9},    & &\frac{178}{35},\\
\frac{1132}{178},& &\frac{293}{27},
\end{eqnarray}
and that \(m_{k + 1} - m_{k} = 1\) is true at the very least, for $\frac{178}{35}$, since \(m_{5} = 3, m_{4} = 2\) when \(k = 4\) and for $\frac{293}{27}$, since \(m_{4} = 4, m_{3} = 3\) when \(k = 3\).\\
\indent From Eqn. (19) we obtain
\begin{equation}
\frac{W(r, k + 1)}{W(r, k)} = r^{m_{k + 1} - m_{k}}\sum_{j = 0}^{m_{k + 1} - m_{k}}(\pm 1)^{j}{m_{k + 1} - m_{k} \choose j}\left(\frac{\alpha}{r}\right)^{j}\frac{c_{m_{k + 1}}}{c_{m_{k}}}(1 + o(1)),
\end{equation}
for large $k$ and if \(m_{k + 1} - m_{k} \geq 1\). However we have also from Eqns. (14)--(16) that
\begin{eqnarray}
\frac{W(r, k + 1)}{W(r, k)}&=     &k^{m_{k + 1} - m_{k}}\left(1 + \frac{1}{k}\right)^{m_{k + 1}}\cdot \frac{c_{m_{k + 1}} + \frac{c_{m_{k + 1} - 1}}{k + 1} + \cdots + \frac{c_{0, m_{k + 1}}}{(k + 1)^{m_{k + 1}}}}{c_{m_{k}} + \frac{c_{m_{k} - 1}}{k} + \cdots + \frac{c_{0, m_{k}}}{k^{m_{k}}}}\nonumber\\
                           &=     &r^{m_{k + 1} - m_{k}}\left(1 \pm \frac{\alpha}{r}\right)^{m_{k + 1} - m_{k}}\left(1 + \frac{1}{r \pm \alpha}\right)^{m_{k + 1}}\\
                           &\times&\frac{c_{m_{k + 1}} + \frac{c_{m_{k + 1} - 1}}{k + 1} + \cdots + \frac{c_{0, m_{k + 1}}}{(k + 1)^{m_{k + 1}}}}{c_{m_{k}}\frac{c_{m_{k} - 1}}{k} + \cdots + \frac{c_{0, m_{k}}}{k^{m_{k}}}},
\end{eqnarray}
for any $k$, $r$, such that \(k = r \pm \alpha\), as we already have described this to mean.

\pagebreak

\end{document}